\documentclass[aps,prl,twocolumn,superscriptaddress]{revtex4-2}

\usepackage{array,caption}
\usepackage[section]{placeins}
\usepackage{pgfplots,amsfonts}
\usepackage[breaklinks]{hyperref}
\usepackage{amssymb}
\usepackage{amsmath,amsthm,bm}
\usepackage{cleveref}
\usepackage{subcaption}
\usepackage{mathtools}
\usepackage{tikz}
\hypersetup{citecolor=red,colorlinks=true,urlcolor=blue}
\usepackage{algorithm,setspace}
\usepackage{algpseudocode}
\usetikzlibrary{angles, quotes}
\usetikzlibrary{positioning}
\usepackage{qcircuit}
\usepackage{ragged2e}
\usepackage{float}

\usepackage[normalem]{ulem}
\usepackage{color}
\usepackage{xcolor}

\newcolumntype{P}[1]{>{\centering\arraybackslash}p{#1}}

\definecolor{Ugreen}{HTML}{198a11}

\newtheorem*{theorem}{Theorem}

\definecolor{darkgreen}{rgb}{0.0, 0.5, 0.0}
\usepackage{graphicx}
\usepackage{etoolbox}

\makeatletter
\pretocmd{\includegraphics}{\hspace*{-0.7cm}}{}{}
\makeatother

\begin{document}

\title{Many-Body Time Evolution from a Correlation-Efficient Quantum Algorithm}

\author{Michael Rose}
\affiliation{Department of Chemistry and The James Franck Institute, The University of Chicago, Chicago, Illinois 60637, USA}

\author{David A. Mazziotti}
\email{damazz@uchicago.edu}
\affiliation{Department of Chemistry and The James Franck Institute, The University of Chicago, Chicago, Illinois 60637, USA}

\date{Submitted October 27, 2025}

\begin{abstract}
We introduce the correlation-efficient time-evolution (CETE) algorithm for simulating quantum many-body dynamics. CETE recasts each step of time evolution as a time-independent correlation problem: the ansatz begins from a mean-field single Slater determinant and is then correlated to capture the true time-evolved state. We derive this exact ansatz from a contraction of the time-dependent Schr{\"o}dinger equation onto the space of two electrons. Unlike conventional evolution by sequential short-time propagators, which must both correlate and decorrelate the state as the degree of correlation fluctuates in time, CETE correlates only once. This substantially reduces circuit depth, extending accessible simulation times on near-term quantum devices. We demonstrate the approach by simulating the time evolution of the hydrogen molecule's electronic wavefunction, highlighting the potential for the CETE algorithm to simulate strongly correlated systems on near-term devices.
\end{abstract}

\maketitle

\noindent\textit{Introduction—}Real-time dynamics governs a wide range of significant phenomena including charge and energy transfer~\cite{krausz_attosecond_2009, calegari_ultrafast_2014, kraus_measurement_2015, folorunso_attochemistry_2023, engel_evidence_2007, schouten_exciton-condensate-like_2025, delgado-granados_quantum_2025}, light--matter interactions~\cite{runge_density-functional_1984, onida_electronic_2002, casida_time-dependent_1996, yabana_time-dependent_1996, moitra_accurate_2023}, chemical reactions ~\cite{bowman_femtosecond_1989, potter_femtosecond_1992, pedersen_validity_1994}, protein dynamics~\cite{ lindorff-larsen_how_2011, hudait_hiv-1_2024}, and non-equilibrium phase transitions~\cite{cavalleri_femtosecond_2001, eichberger_snapshots_2010, aoki_nonequilibrium_2014}. Quantum computers promise to make the simulation of these processes tractable~\cite{feynman_simulating_1982, lloyd_universal_1996, childs_theory_2021}; yet this potential remains unrealized. Current devices accumulate errors as circuit depth increases, limiting accessible simulation times. To address this challenge, a broad family of variational and hybrid quantum algorithms has been proposed~\cite{ollitrault_molecular_2021, mclachlan_variational_1964, gurtin_variational_1965, broeckhove_equivalence_1988,  haegeman_time-dependent_2011, li_efficient_2017, yuan_theory_2019, barison_efficient_2021, lee_variational_2022, nys_ab-initio_2024, gentinetta_correcting_2025, mello_clifford_2025, zhang_adaptive_2025}. Central to these approaches is the choice of ansatz, trading exactness for efficiency. Conventional evolution by sequential short-time propagators is exact but inefficient, as the wavefunction must repeatedly correlate and decorrelate as the degree of correlation fluctuates in time \cite{lloyd_universal_1996, ganoe_notion_2024}.

In this Letter, we introduce the correlation-efficient time-evolution (CETE) algorithm, which recasts time evolution as a time-independent correlation problem~\cite{smart_quantum_2021, smart_verifiably_2024, warren_exact_2024, benavides-riveros_quantum_2024, head-marsden_quantum_2021}. At each time step, the CETE ansatz is initialized from a single Slater determinant $|\phi_0(t+\epsilon)\rangle$ and then correlated via a product of two-electron unitaries. This form emerges naturally from the time-dependent contracted Schr{\"o}dinger equation (TDCSE), which we prove in second quantization to be equivalent to the time-dependent Schr{\"o}dinger equation (TDSE)~\cite{nakatsuji_equation_1999, nakatsuji_equation_1976, mazziotti_contracted_1998}. By correlating only once per time step, the CETE ansatz avoids the inefficiency of conventional, sequential evolution. We demonstrate CETE by simulating the electronic dynamics of H$_2$ on ibm\_fez, measuring the one-particle reduced density matrix (1-RDM, $^1 \hat D$) and energy $\langle \hat H\rangle$ at extended simulation times.

\noindent\textit{Theory—}For molecules and materials that contain at most pairwise interactions, the Hamiltonian can be written as a Hermitian two-electron operator
\begin{align}
\hat H = \sum_{pqrs} {}^2K^{pq}_{rs} {\hat a}^\dagger_p {\hat a}^\dagger_q {\hat a}^{}_s {\hat a}^{}_r
\end{align}
\noindent where ${}^2K$ is the two-electron reduced Hamiltonian, containing both the one- and two-electron integral information. This structure allows the time-dependent Schr{\"o}dinger equation (TDSE)
\begin{align}
\left(\frac{d}{dt} + i \hat H \right) |\psi(t)\rangle = 0
\end{align}
\noindent to be contracted onto the space of two-electron operators, yielding the TDCSE
\begin{align}
\langle \psi(t)| {\hat a}^\dagger_p {\hat a}^\dagger_q {\hat a}^{}_s {\hat a}^{}_r \left(\frac{d}{dt} + i \hat H \right)  |\psi(t)\rangle = 0
\end{align}
\noindent In first quantization, Nakatsuji~\cite{nakatsuji_equation_1999} proved that the TDCSE is a necessary and sufficient condition for the TDSE. Here we prove this equivalence in second quantization.

\begin{theorem} The time-dependent Schr{\"o}dinger equation (TDSE) is satisfied if and only if the time-dependent contracted Schr{\"o}dinger equation (TDCSE) is satisfied.
\end{theorem}
\begin{proof}
Assuming the Hamiltonian $\hat H$ contains at most pairwise interactions, the combined action of the Hamiltonian and the time derivative $\frac{d}{dt}$ can be represented as a two-electron operator
\begin{align}
\left( \frac{d}{dt} + i \hat H \right) = \sum_{pqrs} O^{pq}_{rs} {\hat a}^\dagger_p {\hat a}^\dagger_q {\hat a}^{}_s {\hat a}^{}_r
\end{align}
\noindent in which the coefficients $O^{pq}_{rs}$ are anti-Hermitian and ${\hat a}^\dagger_i$ and ${\hat a}^{}_i$ are creation and annihilation operators. The index $i$ denotes the spin and spatial parts of the orbital as well as the time which is discretized into $t$ and $t+\epsilon$ with $\epsilon \rightarrow 0$ to represent the time derivative.

\indent Weighting the terms of the TDCSE by $-\left( O^{pq}_{rs} \right)$ and summing over all indices yields the variance of the TDSE
\begin{align}
\langle \psi(t)| \left(\frac{d}{dt} + i \hat H \right)^\dagger \left(\frac{d}{dt} + i \hat H \right)  |\psi(t)\rangle = 0
\end{align}

\pagebreak

\noindent The quantity $v^\dagger v$ vanishes if and only if $v$ vanishes. Thus, the TDCSE's satisfaction implies the TDSE's satisfaction.

The converse is immediate. If the TDSE is satisfied, its contraction by any two-electron operator $\langle \psi(t)| {\hat a}^\dagger_p {\hat a}^\dagger_q {\hat a}^{}_s {\hat a}^{}_r $ vanishes. Thus, the TDCSE is satisfied.
\end{proof}

Since $\frac{d}{dt} + i\hat H$ is anti-Hermitian, it is sufficient to express the TDCSE as its anti-Hermitian portion
\begin{align}
\langle \psi(t)| \left\{ {\hat a}^\dagger_p {\hat a}^\dagger_q {\hat a}^{}_s {\hat a}^{}_r , \left(\frac{d}{dt} + i \hat H \right) \right\} |\psi(t)\rangle = 0 .
\end{align}
\noindent From this equivalence a natural choice of ansatz emerges. For a small time step $\epsilon$, the evolved wavefunction can be written as an exponential of the time derivative
\begin{align}
    |\psi(t + \epsilon)\rangle = e^{\epsilon\frac{d}{dt}} |\psi(t) \rangle
\end{align}
\noindent To satisfy the TDSE, the action of $\frac{d}{dt}$ must cancel that of $i\hat H$. Consequently, the action of $\frac{d}{dt}$ can be minimally represented as an anti-Hermitian two-electron operator $i\hat S$; thus, the ansatz takes the form of a two-body unitary operator acting on the wavefunction at time $t$:
\begin{align}
|\psi(t + \epsilon)\rangle= e^{i \epsilon \hat S} |\psi(t)\rangle
\end{align}
\noindent Because the tangent space of this ansatz coincides with the two-electron test functions used in the contraction of the TDCSE, its stationarity ensures that the TDCSE—and therefore the TDSE—is satisfied.

For evolution over longer times, the ansatz takes the form of the product of two-electron unitaries for each $\epsilon$ time step. However, in practice, the depth of this ansatz can be reduced by correlating from a single Slater determinant $|\phi_0\rangle$, rather than starting from the previous wavefunction $|\psi(t)\rangle$. This yields an ansatz of the form
\begin{align}
|\phi_{m+1}(t)\rangle = \left( \prod_m e^{\,i\epsilon_m \hat S_m(t)} \right) |\phi_0\rangle .
\end{align}
\noindent With this substitution, the ansatz depth is bounded by the degree of correlation at the target time rather than the number of time steps; on noise-limited devices this decreases the effect of noise and enables longer evolution. Moreover, by retaining the two-electron unitary structure, this ansatz remains exact. This wavefunction $|\phi_{m+1}(t)\rangle$ can be efficiently prepared on a quantum device, by initializing the single Slater determinant $|\phi_0\rangle$ and then implementing each two-electron unitary via a Trotterized product formula~\cite{suzuki_generalized_1976, lloyd_universal_1996, childs_theory_2021}.

To perform time evolution on a quantum device, we propose the correlation-efficient time-evolution (CETE) algorithm (Table~\ref{tab:cqte}). An implementation of this algorithm is openly available \cite{rose_correlation-efficient-time-evolution-cete_nodate}. Starting at time $t$, we create a target wavefunction $|\chi(t+\epsilon)\rangle $ by evolving $|\psi(t)\rangle$ a step $\epsilon$ forward
\begin{align}
|\chi(t+\epsilon)\rangle = e^{-i\hat H\epsilon} |\psi(t)\rangle .
\end{align}
The practical selection of $\epsilon$ depends on several factors including the strength and type of noise on the quantum device and the time dependence of the Hamiltonian.

\begin{table}[t]
\caption{\justifying\raggedright CETE algorithm for the time evolution of systems with at most pairwise interactions}
\label{tab:cqte}
\begin{ruledtabular}
\begin{tabular}{l}
\textbf{CETE Algorithm} \\
\textbf{From} $t=t_0$ to $t=t_{\max}$: \\[2pt]
\quad Initialize $|\chi(t+\epsilon)\rangle = e^{-i\hat H\epsilon}|\psi(t)\rangle $ \\[2pt]
\quad Set $|\phi_0(t+\epsilon)\rangle \leftarrow |\phi_0\rangle$ \\[2pt]
\quad  Set $m \leftarrow 0$ \\[2pt]
\quad \textbf{While} $F_m < 1-\delta_{\rm{cutoff}}$ \textbf{:} \\
\quad\quad \textbf{1:} Evaluate $i\hat S_{m}$ \\
\quad\quad \textbf{2:} Update $|\phi_{m+1}(t+\epsilon)\rangle
= e^{i\epsilon_m \hat S_m}\,|\phi_m(t+\epsilon)\rangle$ \\
\quad\quad \textbf{3:} $m \leftarrow m+1$ \\
\quad Set $|\psi(t+\epsilon)\rangle \leftarrow |\phi_m(t+\epsilon)\rangle $ \\
\quad Set $t \leftarrow t+\epsilon$ \\
\end{tabular}
\end{ruledtabular}
\end{table}

\indent An ansatz for this time $|\phi_m(t+\epsilon) \rangle $ is then constructed iteratively through the action of two-electron unitaries. The $0^{\rm th}$ iteration ansatz is a single Slater determinant $|\phi_0\rangle $. The reference wavefunction $|\phi_0\rangle$ can be kept fixed for all time or updated at each time by selecting the Slater determinant that is most probable at the current time
\begin{align}
|\phi_{m+1}(t+\epsilon)\rangle = \left( \prod_m e^{\,i\epsilon_m \hat S_m(t+\epsilon)} \right) |\phi_0\rangle .
\end{align}
\noindent  These two-electron unitaries are constructed iteratively by setting each residual $i\hat S_{m+1}$ equal to the gradient of the fidelity between the target state and ansatz (see Supplemental Material)~\cite{barison_efficient_2021}
\begin{align}
iS^{pq;rs}_{m+1} \leftarrow \frac{d F_{m+1}}{diS^{pq;rs}_{m+1}}
\end{align}
\noindent The fidelity can be measured by first preparing $|\chi(t+\epsilon)\rangle$ on the device, then applying the inverse of each unitary in the ansatz, and finally measuring the probability of $|\phi_0\rangle$
\begin{align}
F_{m+1} = || \langle \phi_{m+1}(t+\epsilon) |\chi(t+\epsilon)\rangle ||^2 .
\end{align}
After the fidelity exceeds a predetermined cutoff, the ansatz $|\phi_{m+1}(t+\epsilon)\rangle$ becomes the definition of the wavefunction $|\psi(t+\epsilon)\rangle$ for further time evolution to $t+2\epsilon$ and so on. In the limit that the fidelity reaches unity, the ansatz reproduces exact time evolution. If the gradient vanishes before the cutoff is reached, sequential short-time propagators can be used for a time step.

\pagebreak

\begin{figure}[t]
    \centering
    \includegraphics[width=\columnwidth]{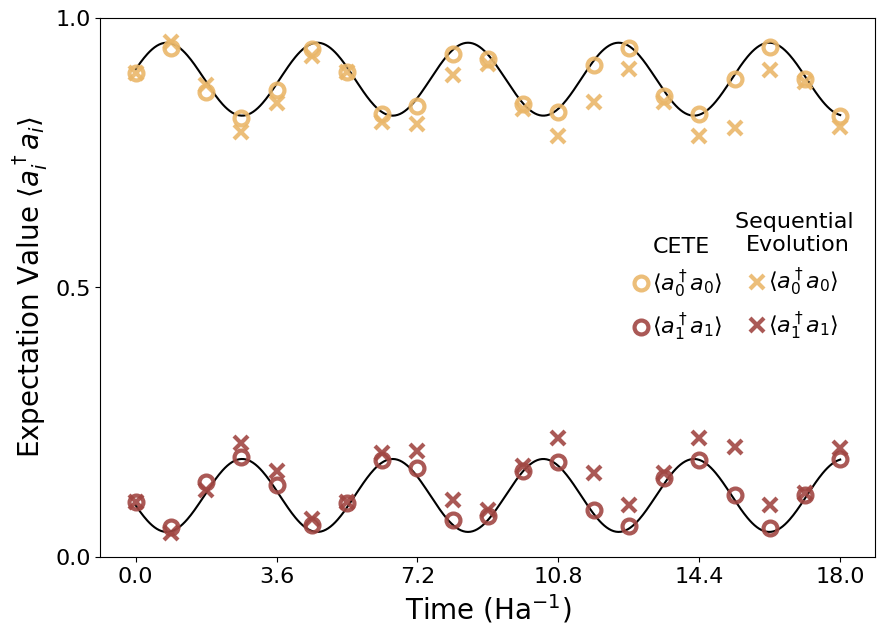}
    \caption{\justifying Measured diagonal 1-RDM $^1\hat D$ elements for H$_2$ obtained using the CETE ansatz and sequential evolution as a function of time. Pauli-sum tomography is performed at intervals of $0.9~\text{Ha}^{-1}$ using $10^4$ shots per Pauli string. All measurements are made on qubit 137 of ibm\_fez. The solid line indicates a noiseless state-vector reference. At all time points $\langle {\hat a}^\dagger_0 {\hat a}^{}_0 \rangle = \langle a_2^\dagger {\hat a}^{}_2 \rangle$ and $\langle {\hat a}^\dagger_1 {\hat a}^{}_1 \rangle = \langle {\hat a}^\dagger_3 {\hat a}^{}_3 \rangle$.}
    \label{fig:H2_semireal}
\end{figure}

\begin{figure}[t]
    \centering
    \includegraphics[width=\columnwidth]{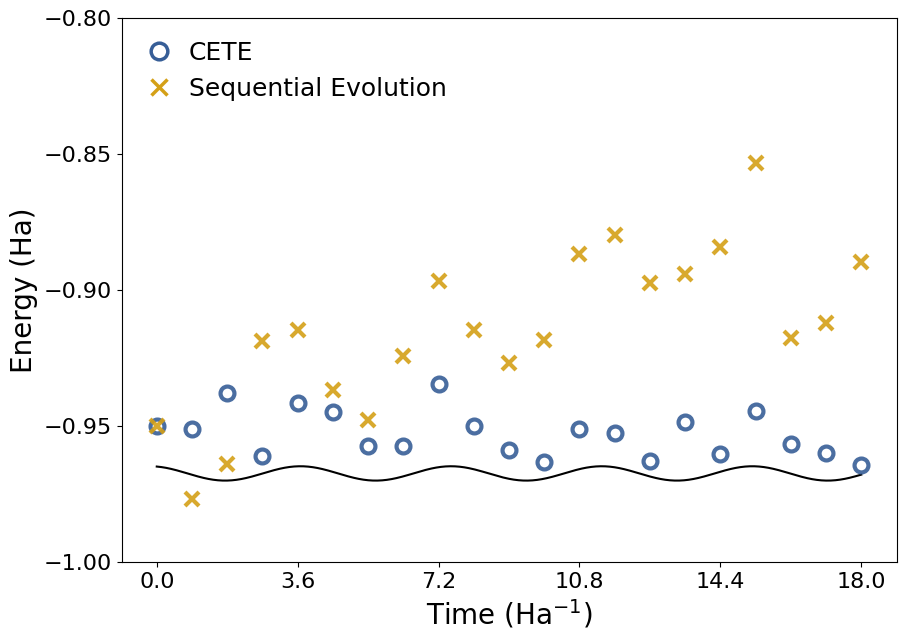}
    \caption{\justifying Measured energy values of H$_2$ obtained using the CETE ansatz $\boldsymbol{\circ}$ and sequential evolution $\boldsymbol{\times}$ as a function of time. Pauli-sum tomography is performed at intervals of $0.9~\text{Ha}^{-1}$ with $10^4$ shots per Pauli string. All measurements are made on qubit 137 of ibm\_fez. The solid line indicates a noiseless state-vector reference.}
    \label{fig:H2_energy}
\end{figure}

\noindent\textit{Applications—}We illustrate the CETE algorithm by simulating the time evolution of H$_2$ at a fixed bond distance of $0.735~\text{\AA}$ in the Slater-type orbital (STO-3G) basis set \cite{hehre_self-consistent_1969}. The electronic wavefunction is evolved for a total of $18~\text{Ha}^{-1}$ (435.4 as) using time steps of $0.90~\text{Ha}^{-1}$ that are Trotterized into substeps of $0.03~\text{Ha}^{-1}$. The initial state is prepared as
\begin{align}
|\psi(t_0)\rangle = e^{0.1 \pi \hat A}|0101\rangle ,
\end{align}
where $\hat A = i \left({\hat a}^\dagger_1 {\hat a}^\dagger_3 {\hat a}^{}_2 {\hat a}^{}_0 + {\hat a}^\dagger_0 {\hat a}^\dagger_2 {\hat a}^{}_3 {\hat a}^{}_1\right)$ is a double excitation from the ground state. Slater determinants correspond to the occupations of Hartree--Fock (HF) molecular orbitals and are encoded in spin-block, little-endian order. Because only the HF ground state $|0101\rangle$ and the double excitation $|1010\rangle$ are coupled during evolution, the problem reduces to a single qubit from the mappings $|0101\rangle \mapsto |0\rangle$ and $|1010\rangle \mapsto |1\rangle$. The 1-RDM ($^1D^i_j = \langle \psi|{\hat a}^\dagger_i{\hat a}^{}_j|\psi\rangle$), two-particle reduced density matrix (2-RDM,$^2D^{ij}_{kl} = \langle \psi|{\hat a}^\dagger_i{\hat a}^\dagger_j{\hat a}^{}_l{\hat a}^{}_k|\psi\rangle$) and Hamiltonian $\hat H$ are similarly reduced to the space of a single qubit. Restricting the mapping to Slater determinants with identical $\alpha$ and $\beta$ blocks enforces $\langle {\hat a}_0^\dagger {\hat a}^{}_0\rangle=\langle {\hat a}_2^\dagger {\hat a}^{}_2\rangle$ and $\langle {\hat a}_1^\dagger {\hat a}^{}_1\rangle=\langle {\hat a}_3^\dagger {\hat a}^{}_3\rangle$ symmetries in the diagonal of the 1-RDM where each diagonal element represents the probability of finding an electron in the given spin orbital.

For each $0.90~\text{Ha}^{-1}$ time step, the CETE ansatz is constructed on the noiseless AerSimulator \cite{aleksandrowicz_qiskit_2019}. The HF ground state is employed as the single Slater determinant reference state $|\phi_0\rangle$. The two-electron unitaries are implemented by mapping operators to Pauli sums under the Jordan--Wigner transformation, and the gradients are calculated analytically using $2\times10^4$ shots per term (see Supplemental Material)~\cite{brandwood_complex_1983, mitarai_quantum_2018, schuld_evaluating_2019, anselmetti_local_2021, kottmann_feasible_2021, wierichs_general_2022}. \\

\begin{figure}[b]
    \centering
    \includegraphics[width=\columnwidth]{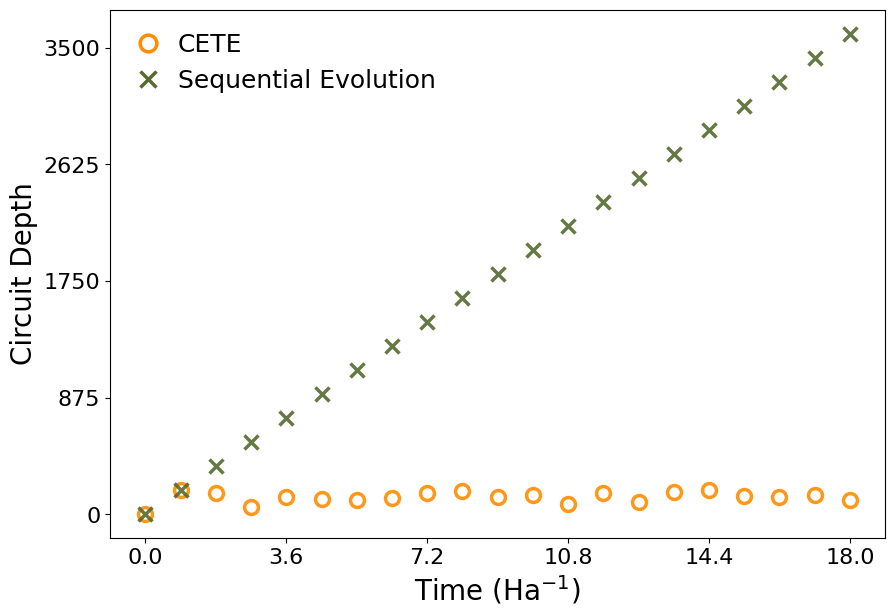}
    \caption{\justifying Depth of quantum circuits implementing the CETE ansatz $\boldsymbol{\circ}$ and sequential evolution $\boldsymbol{\times}$ as a function of time. Depths are evaluated at $0.9~\text{Ha}^{-1}$ intervals. The CETE ansatz for each time point is constructed using the noiseless AerSimulator.}
    \label{fig:H2_depth}
\end{figure}

To assess performance, we compare the CETE circuits against those produced by sequential short-time propagators. Pauli-sum tomography of the 1-RDM and energy are performed for both the ansatz and sequential evolution circuits,as shown in Figs.~\ref{fig:H2_semireal} and \ref{fig:H2_energy}, respectively. Qubit 137 on ibm\_fez is selected for its low readout error at the time of calculation. The consistently shallower depths of CETE circuits, as shown in Fig.~\ref{fig:H2_depth}, significantly suppress the effects of noise. The average and peak absolute errors of diagonal elements in the 1-RDM are reduced by factors exceeding four and five, respectively. Similarly, the average and peak absolute errors in the energy expectation values are reduced by factors greater than three. This reduction in depth arises because the CETE ansatz correlates the state once, rather than repeatedly correlating and decorrelating as the degree of correlation fluctuates in time.

\noindent\textit{Conclusions—}By reframing each step of time evolution as a time-independent correlation problem, CETE marks a paradigm shift: it retains exactness while bounding circuit depth by the degree of correlation. This correlation-bounded scaling extends simulation times and is essential for reaching nuclear time scales that are orders of magnitude longer than those of the electronic degrees of freedom. These time scales are necessary to capture the bond breaking and forming of chemical reactions and phase transitions. As this algorithm derives from the TDCSE-TDSE equivalence, it is universal for all Hamiltonians with at most pairwise interactions, enabling applications across a broad range of phenomena in both molecules and materials.

\begin{acknowledgments}
D. M. gratefully acknowledges the U.S. Department of Energy, Office of Science, Basic Energy Sciences under Award Number DE-SC0026076 for support.  We acknowledge the use of IBM Quantum services for this work. The views expressed are those of the authors, and do not reflect the official policy or position of IBM or the IBM Quantum team.
\end{acknowledgments}

\bibliography{abbrev, references_abbrev}

\pagebreak

\end{document}